\documentclass[12pt]{article}
\usepackage{a4wide}
\usepackage{amsthm}
\usepackage{amsmath}
\usepackage{amssymb}
\def\CC{{\cal C}}
\def\TT{{\cal T}}
\def\FFF{{\cal F}}
\def\FF{{\mathbb F}}
\def\GF{{\mathbb G\mathbb F}}
\def\GG{{\cal G}}
\def\II{{\cal I}}
\def\MM{{\cal M}}
\def\NN{{\mathbb N}}
\def\div{\;\mbox{div}\;}
\def\mod{\;\mbox{mod}\;}
\def\indep{{\rm Indep}}
\newtheorem{theorem}{Theorem}
\newtheorem{proposition}[theorem]{Proposition}
\newtheorem{lemma}[theorem]{Lemma}
\newtheorem{corollary}[theorem]{Corollary}
\begin{document}
\title{Deciding First Order Logic Properties of Matroids}
\author{Tom{\'a}{\v s} Gaven{\v c}iak\thanks{Department of Applied Mathematics, Faculty of Mathematics and Physics, Charles University, Malostransk{\'e} n{\'a}m{\v e}st{\'\i} 25, 118 00 Prague, Czech Republic. E-mail: {\tt gavento@kam.mff.cuni.cz}. This research was supported by the grant GAUK 64110.}\and
        Daniel Kr{\'a}l'\thanks{Department of Mathematics, University of West Bohemia, Univerzitn\'i 8, 306 14 Pilsen, Czech Republic. This author is also affiliated with Department of Applied Mathematics and Institute for Theoretical Computer Science of Charles University. E-mail: {\tt dankral@kma.zcu.cz}. This author was supported by the grant GA{\v C}R P202/11/0196.}\and
	Sang-il Oum\thanks{Department of Mathematical Sciences, KAIST, 291 Daehak-ro Yuseong-gu, Daejeon 305-701, South Korea. E-mail: {\tt sangil@kaist.edu}. This research was supported by Basic Science Research Program through the National Research Foundation of Korea(NRF) funded by the Ministry of Education, Science and Technology(2011-0011653) and also by TJ Park Junior Faculty Fellowship.}}
\date{}
\maketitle
\begin{abstract}
Frick and Grohe [J. ACM 48 (2006), 1184--1206] introduced a notion
of graph classes with locally bounded tree-width and established that
every first order logic property can be decided in almost linear time
in such a graph class. Here, we introduce an analogous
notion for matroids (locally bounded branch-width) and show
the existence of a fixed parameter algorithm for first order logic properties
in classes of regular matroids with locally bounded branch-width.
To obtain this result, we show that the problem of deciding the existence
of a circuit of length at most $k$ containing two given elements is fixed
parameter tractable for regular matroids.
\end{abstract}

\section{Introduction}

Classes of graphs with bounded tree-width play
an important role both in structural and algorithmic graph theory.
Their importance from the structural point of view stems
from their close relation to graph minors. One of the main results
in this area asserts that graphs in a minor-closed class of graphs
admit a tree-like decomposition into graphs almost embeddable on surfaces.
In particular, a minor-closed class of graphs
has bounded tree-width if and only if one of the forbidden minors for the class
is planar~\cite{bib-minorsV}.
With respect to algorithms, there have been an enormous amount of reports
on designing efficient algorithms for classes of graphs with bounded tree-width
for various NP-complete problems. Most of such algorithms have been put into a uniform
framework by Courcelle~\cite{bib-courcelle90} who showed that every monadic second order
can be decided in linear-time for a class of graphs with bounded tree-width.
The result can also be extended
to computing functions given in the monadic second order logic~\cite{bib-courcelle97}.

Not all minor-closed graph classes have bounded tree-width and though
some important graph properties can be tested in linear or almost linear time
for graphs from these classes.
One example of such property is whether a graph has an independent or a dominating set of
size at most $k$. Clearly, such a property can be tested in polynomial time for a fixed $k$.
However, for the class of planar graphs, these properties could be tested
in linear time for any fixed $k$.
More generally, if a class of graphs have locally bounded tree-width (e.g., this
holds the class of planar graphs),
the presence of a subgraph isomorphic to a given fixed subgraph $H$ can be tested
in linear time~\cite{bib-eppstein1,bib-eppstein2}.
Frick and Grohe~\cite{bib-frick+} developed a framework unifying such algorithms.
They showed that every first order logic property can be tested in almost linear
time for graphs in a fixed class of graphs with locally bounded tree-width. Here, almost linear
time stands for the existence of an algorithm $A_\varepsilon$ running
in time $O(n^{1+\varepsilon})$ for every $\varepsilon>0$. More general results
in this direction apply to
classes of graphs locally excluding a minor~\cite{bib-dawar} or
classes of graphs with bounded expansion~\cite{bib-our-focs}.

Matroids are combinatorial structures generalizing graphs and linear independance.
The reader is referred
to Section~\ref{sect-notation} if (s)he is not familiar with this concept. For the exposition now,
let us say that every graph $G$ can be associated with a matroid $M(G)$ that captures
the structure of cycles and cuts of $G$. Though the notion of tree-width can be generalized
from graphs to matroids~\cite{bib-hlineny+,bib-hlineny+err}, a more natural decomposition parameter
for matroids is the branch-width which also exists for graphs~\cite{bib-minorsX}.
The two notions are closely related:
if $G$ is a graph with tree-width $t$ and its matroid $M(G)$ has branch-width $b$,
then it holds that $b\le t+1\le 3b/2$~\cite{bib-minorsX}.

Hlin{\v e}n{\'y}~\cite{bib-hlineny03,bib-hlineny06}
established an analogue of the result of Courcelle~\cite{bib-courcelle90} for classes of matroids
with bounded branch-width that are represented over a fixed finite field. In particular,
he showed that every monadic second order logic property can be decided in cubic time
for matroids from such a class which are given by their representation over the field (the decomposition
of a matroid needs not be a part of the input as it can be efficiently computed,
see Subsection~\ref{subsect-bw} for more details). Since any graphic matroid, i.e.,
a matroid corresponding to a graph, can be represented over any (finite) field and
its representation can be easily computed from the corresponding graph, Hlin{\v e}n{\'y}'s
result generalizes Courcelle's result to matroids.

The condition that a matroid is given by its representation seems to be necessary
as there are several negative results on deciding representability of matroids:
Seymour~\cite{bib-seymour81} argued that it is not possible to test representability
of a matroid given by the independance oracle in subexponential time. His construction
readily generalizes to every finite field and it is possible to show that the matroids
appearing in his argument have branch-width at most three.
If the input matroid is represented over rationals,
it can be tested in polynomial-time whether it is binary~\cite{bib-seymour80}
even without a constraint on his branch-width, but
for every finite field $\FF$ of order four or more and every $k\ge 3$,
deciding representability over $\FF$ is NP-hard~\cite{bib-hlineny06-mfcs}
for matroids with branch-width at most $k$ given by their representation
over rationals.

In this paper, we propose a notion of locally bounded branch-width and
establish structural and algorithmic results for such classes of matroids.
In particular, we aim to introduce an analogue of the result
of Frick and Grohe~\cite{bib-frick+} for matroids in the spirit how the result
of Hlin{\v e}n{\'y}~\cite{bib-hlineny06} is an analogue of that of Courcelle~\cite{bib-courcelle90}.
Our main result (Theorem~\ref{thm-main}) asserts that testing first order logic
properties in classes of regular matroids
with locally bounded branch-width
that are given by the independance oracle
is fixed parameter tractable.

\begin{table}
\begin{center}
\begin{tabular}{|p{3cm}|p{4.5cm}|p{4.5cm}|}
\hline
Properties & Graphs & Matroids \\
\hline
First order logic & bounded tree-width~\cite{bib-courcelle90} & bounded branch-width and represented over a~finite field~\cite{bib-hlineny06}\\
\hline
Monadic second order logic & locally bounded tree-width~\cite{bib-frick+} & {\bf locally bounded branch-width and regular}\\
\hline
\end{tabular}
\end{center}
\caption{Fixed parameter algorithms for testing first order and monadic second order logic
         properties for classes of graphs and matroids. Our result is printed in bold font.}
\label{table}
\end{table}

The first issue to overcome on our way is that the definition of graph
classes with locally bounded tree-width is based on distances between vertices in graphs,
a concept missing (and impossible to be analogously defined) in matroids.
Hence, we introduce a different sense of locality based on circuits of matroids and
show in Section~\ref{sect-ltw+lbw} that if a class of graphs has locally bounded tree-width,
then the corresponding class of matroids has locally bounded branch-width,
justifying that our notion is appropriately defined.
We then continue with completing the algorithmic picture (see Table~\ref{table})
by establishing that every first order logic property is fixed parameter tractable
for classes of regular matroids with locally bounded branch-width. Recall that
all graphic matroids are regular. Let us remark that
we focus only on establishing the fixed parameter tractability of the problem
without attempting to optimize its running time.

In one step of our algorithm, we need to be able to decide the existence of a circuit
containing two given elements with length at most $k$ in a fixed parameter
way (when parametrized by $k$).
However, the Maximum likelihood decoding problem for binary inputs,
which is known to be W[1]-hard~\cite{bib-downey}, can be reduced
to this problem (in the Maximum likelihood decoding problem,
one is given a target vector $v$ and a matrix $M$ and the task
is to decide whether there exists a vector $v'$ with at most $k$ non-zero coordinates
such that $v=Mv'$). Hence, it is unlikely that the existence of a circuit
with length at most $k$ containing two given elements is fixed parameter
tractable. So, we had to restrict
to a subclass of binary matroids. We have chosen regular matroids and
established the fixed parameter tractability of deciding the existence of
a circuit containing two given elements with length at most $k$ in this
class of matroids (Theorem~\ref{thm-circuit}). We think that this result
might be of independent interest.

\section{Notation}
\label{sect-notation}

In this section, we introduce notation used in the paper.

\subsection{Tree-width and local tree-width}
\label{subsect-tw}

A {\em tree-decomposition} of a graph $G$ with vertex set $V$ and edge set $E$
is a tree $T$ such that every node $x$ of $T$ is assigned a subset $V_x$ of $V$.
To avoid confusion, we always refer to vertices of graphs as to {\em vertices} and
those appearing in tree-decompositions or branch-decompositions as to {\em nodes}.
The tree-decomposition must satisfy that
\begin{itemize}
\item for every edge $vw\in E$, there exists a node $x$ such that $\{v,w\}\subseteq V_x$, and 
\item for every vertex $v\in V$, the nodes $x$ with $v\in V_x$ induce a subtree of $T$.
\end{itemize}
The second condition can be reformulated that for any two nodes $x$ and $y$, all the vertices
of $V_x\cap V_y$ are contained in the sets $V_z$ for all nodes $z$ on the path between $x$ and $y$ in $T$.
The {\em width} of a tree-decomposition is $\max_{x\in V(T)}|V_x|-1$ and
the {\em tree-width} of a graph $G$ is the minimum width of its tree-decomposition.
It can be shown that a graph has tree-width at most one if and only if it is a forest.
A class $\GG$ of graphs has {\em bounded tree-width} if there exists an integer $k$
such that the tree-width of any graph in $\GG$ is at most $k$.

The set of vertices at distance at most $d$ from a vertex $v$ of $G$ is denoted by $N^G_d(v)$ and
it is called the {\em $d$-neighborhood} of $v$.
A class $\GG$ of graphs has {\em locally bounded tree-width} if there exists a function $f:\NN\to\NN$ such that
the tree-width of the subgraph induced by the vertices of $N^G_d(v)$ is at most $f(d)$ for every graph $G\in\GG$,
every vertex $v$ of $G$ and every non-negative integer $d$. Unlike the notion of tree-width,
there is no notion of local tree-width for a single graph and the notion is defined for classes of graphs.

\subsection{Matroids}
\label{subsect-matroid}

We briefly introduce basic definitions related to matroids;
further exposition on matroids can be found in several monographs on the subject,
e.g., in~\cite{bib-oxley,bib-truemper}.
A {\em matroid} is a pair $(E,\II)$ where $E$ are the elements of $M$ and $\II$ is a family of subsets of $E$
whose elements are referred as {\em independent} sets and
that satisfies the following three properties:
\begin{enumerate}
\item the empty set is contained in $\II$,
\item if $X$ is contained in $\II$, then every subset of $X$ is also contained in $\II$, and
\item if $X$ and $Y$ are contained in $\II$ and $|X|<|Y|$, then there exists $y\in Y$ such that $X\cup\{y\}$
      is also contained in $\II$.
\end{enumerate}
The set $E$ is called the {\em ground set} of $M$.
Clearly, if $M=(E,\II)$ is a matroid and $E'\subseteq E$,
then $(E',\II\cap 2^{E'})$ is also a matroid; we say that this matroid is the matroid $M$ {\em restricted} to $E'$.

The maximum size of an independent subset of a set $X$ is called
the {\em rank} of $X$ and it can be shown that all inclusion-wise maximum subsets of $X$ have the same size.
Clearly, a set $X$ is independent if and only if its rank is $|X|$.
The rank function of a matroid, i.e., the function $r:2^E\to\NN\cup\{0\}$ assigning a subset $X$ its rank, is submodular
which means that $r(A\cap B)+r(A\cup B)\le r(A)+r(B)$ for any two subsets $A$ and $B$ of $E$.
The inclusion-wise maximal subsets contained in $\II$ are called {\em bases}.
An element $e$ of $M$ such that $r(\{e\})=0$ is called a {\em loop};
on the other hand, if $r(E\setminus\{e\})=r(E)-1$, then $e$ is a {\em coloop}.

Two most classical examples of matroids are vector matroids and graphic matroids.
In {\em vector matroids},
the elements $E$ of a matroid are vectors and a subset $X$ of $E$ is independent if the vectors in $X$
are linearly independent. The dimension of the linear hull of the vectors in $X$ is equal to the rank of $X$.
We say that a matroid $M$ is {\em representable} over a field $\FF$
if it is isomorphic to a vector matroid over $\FF$;
the corresponding assignment of vectors to the elements of $M$
is called a {\em representation} of $M$ over $\FF$.
A matroid representable over the two-element field is called {\em binary}.

If $G$ is a graph, its {\em cycle matroid $M(G)$} is the matroid with ground set formed by the edges of $G$ and
a subset $X$ is independent in $M(G)$ if $G$ does not contain a cycle formed by some edges of $X$.
A matroid is {\em graphic} if it is a cycle matroid of a graph.
Dually, the {\em cut matroid $M^*(G)$} is the matroid with the same ground set as $M(G)$
but a subset $X$ is independent if the graphs $G$ and $G\setminus X$ have the same number of components.
The matroid $M^*(G)$ is sometimes called the bond matroid of $G$. Cut matroids are called {\em cographic}.

We now introduce some additional notation related to matroids.
Inspired by a graph terminology, a {\em circuit} of a matroid $M$ is a subset $X$ of the elements of $M$
that is an inclusion-wise minimal subset that is not independent.
The following is well-known.

\begin{proposition}
\label{prop-2circuits}
Let $C_1$ and $C_2$ be two circuits of a matroid with $C_1\cap C_2\not=\emptyset$ and
$e$ and $f$ two distinct elements in $C_1\cup C_2$. There exists a circuit $C\subseteq C_1\cup C_2$
such that $\{e,f\}\subseteq C$.
\end{proposition}

If $M=(E,\II)$ is a matroid, then the dual matroid $M^*$ of $M$ is the matroid with ground set $E$ such that
a set $X$ is independent in $M^*$ if and only if the complement of $X$ contains a base of $M$.
A circuit of $M^*$ is referred to as a {\em cocircuit} of $M$.
Dual matroids of graphic matroids are cographic and vice versa. 

A {\em $k$-separation} of a matroid $M$ is a partition of its ground set into two sets $A$ and $B$ such
that $r(A)+r(B)-r(M)=k-1$. A matroid $M$ is {\em connected} if it has no $1$-separation,
i.e., its ground set cannot be partitioned into two non-empty sets $A$ and $B$ with $r(A)+r(B)=r(M)$.
This is equivalent to the statement that for every pair $e_1$ and $e_2$ of elements of $M$,
there exists a circuit in $M$ that contains both $e_1$ and $e_2$.
The notions of connectivity for matroids and graphs are different.
A matroid $M(G)$ corresponding to a graph $G$ is connected if and only if
$G$ is $2$-connected~\cite{bib-oxley}.

We now explain how a binary matroid can be obtained from two smaller binary matroids
by gluing along a small separation. Before doing so, we need one more definition:
a {\em cycle} is a disjoint union of circuits in a binary matroid.
Let $M_1$ and $M_2$ be two binary matroids with ground sets $E_1$ and $E_2$, respectively.
Suppose that $\min\{|E_1|,|E_2|\}\le |E_1\triangle E_2|$ and that one of the following holds.
\begin{itemize}
\item $E_1$ and $E_2$ are disjoint.
\item $E_1\cap E_2$ contains a single element which is neither a loop nor a coloop in $M_1$ or $M_2$.
\item $E_1\cap E_2$ is a $3$-element circuit in both $M_1$ and $M_2$ and $E_1\cap E_2$ contains no cocircuit in $M_1$ or $M_2$.
\end{itemize}
We define $M=M_1\triangle M_2$ to be the binary matroid with ground set $E_1\triangle E_2$
such that a set $C\subseteq E_1\triangle E_2$ is a cycle in $M$ if and only if there
exist cycles $C_i$ in $M_i$, $i=1,2$, such that $C=C_1\triangle C_2$. Based on which of the three
cases apply, we say that $M$ is a {\em $1$-sum}, a {\em $2$-sum} or a {\em $3$-sum} of $M_1$ and $M_2$.
Note that if $M$ is a $k$-sum, $k=1,2,3$, of $M_1$ and $M_2$,
then the partition $(E_1\setminus E_2,E_2\setminus E_1)$ is a $k$-separation of $M$.

A matroid $M$ is {\em regular} if it is representable over any field.
It is well-known that a matroid $M$ is regular if and only if
it has a representation with a totally unimodular matrix (over rationals).
Seymour's Regular Matroid Decomposition Theorem~\cite{bib-seymour80}
asserts that every regular matroid can be obtained by a series of $1$-,
$2$- and $3$-sums of graphic matroids, cographic matroids and
copies of a particular $10$-element matroid referred as $R_{10}$.

\subsection{Matroid algorithms}
\label{subset-alg}

In this paper, we study efficient algorithms for matroids. Since this area involves several
surprising aspects, let us provide a brief introduction.
Since the number of non-isomorphic $n$-element matroids is double exponential in $n$,
an input matroid cannot be presented to an algorithm using subexponential space.
So, one has to resolve how an input matroid is given.
One possibility is that the input matroid is given by an oracle (a black-box function) that
returns whether a given subset of the ground set is independent.
If this is the case, we say that the matroid is {\em oracle-given}.
Other possibilities include giving the matroid by its representation over a field.
Since algorithms working with oracle-given matroids are the most general,
we assume in this paper that an input matroid is oracle-given unless stated otherwise.
The running time of algorithms is usually measured in the number of elements of an input matroid.
So, if the matroid is graphic, the running time is measured in the number of its edges.

Let us now mention some particular algorithmic results, some of which will be used later.
Though it is easy to construct a representation over $\GF(2)$ for an oracle-given binary matroid 
in quadratic time (under the promise that the input matroid is binary),
deciding whether an oracle-given matroid is binary cannot be solved in subexponential time~\cite{bib-seymour81}.
However, it can be decided in polynomial time whether an oracle-given matroid is graphic~\cite{bib-seymour81}.
Moreover, if the matroid is graphic, the corresponding graph can be found efficiently,
see, e.g., \cite[Chapter 10]{bib-truemper} for further details.

\begin{theorem}
\label{thm-graphicness}
There is a polynomial-time algorithm that decides whether an oracle-given matroid $M$ is graphic and,
if so, finds a graph $G$ such that $M=M(G)$.
\end{theorem}

Similarly, one cast test whether a given matroid is regular (also
see~\cite{bib-truemper-jctb} or \cite[Chapter 11]{bib-truemper}).

\begin{theorem}[Seymour~\cite{bib-seymour80}]
\label{thm-regular}
There is a polynomial-time algorithm that decides whether an oracle-given matroid $M$ is regular and,
if so, outputs a sequence containing graphic matroids, cographic matroids and the matroid $R_{10}$
along with the sequence $1$-, $2$- and $3$-sums yielding $M$.
\end{theorem}

\subsection{Branch-width and local branch-width}
\label{subsect-bw}

Though the notion of branch-width is also defined for graphs, we will use it exclusively for matroids.
A {\em branch-decomposition} of a matroid $M=(E,\II)$ is a subcubic tree $T$
with leaves one-to-one corresponding to the elements of $M$. Each edge $e$ of $T$
splits $T$ into two subtrees and thus naturally partitions the elements of $M$ into two subsets $A$ and $B$,
those corresponding to the leaves of the two subtrees. The {\em width of the edge $e$} is $r(A)+r(B)-r(M)+1$
where $r$ is the rank function of $M$, and the {\em width of the branch-decomposition} is the maximum
width of its edge. The {\em branch-width} of $M$ is the minimum width of its branch-decomposition.

As mentioned earlier, the branch-width of the cycle matroid $M(G)$ of a graph $G$ is closely
related to the tree-width of $G$.
\begin{proposition}[\cite{bib-minorsX}]
\label{prop-bw+tw}
Let $G$ be a graph.
If $b$ is the branch-width of the cycle matroid $M(G)$ of $G$ and $t$ is the tree-width of $G$,
then $ b-1 \le t \le 3b/2$.
\end{proposition}

Similarly to graphs, the branch-width of a matroid cannot be efficiently computed.
However, for a fixed $k$, testing whether an oracle-given matroid
has branch-width at most $k$ can be done in polynomial time~\cite{bib-oum+}.
For matroids represented over a finite field, a cubic-time approximation algorithm
for the branch-width of a given matroid (for every fixed branch-width)
have been designed in~\cite{bib-hlineny05} and
the result have later been extended to a cubic-time algorithm deciding whether
a given matroid has branch-width at most $k$:
\begin{theorem}[Hlin{\v e}n{\'y} and Oum~\cite{bib-hlineny+oum}]
\label{thm-hlineny+oum}
Let $k$ be a fixed integer and $\FF$ a fixed finite field. There is a cubic-time algorithm that
decides whether the branch-width of a matroid $M$ represented over $\FF$ is at most $k$ and,
if so, it constructs a branch-decomposition of $M$ with width at most $k$.
\end{theorem}

One obstacle that needs to be overcome when defining local branch-width is the absence of 
the natural notion of distance for matroids. We overcome this in a way similar to that
used in the definition of matroid connectivity from Subsection~\ref{subsect-matroid},
i.e., we utilize containment in a common circuit.
Let $M$ be a matroid. The {\em distance} of two elements $e$ and $f$ of $M$
is equal to the minimum size of a circuit of $M$ containing both $e$ and $f$;
if $e=f$, then the distance of $e$ and $f$ is equal to zero.
Proposition~\ref{prop-2circuits} implies that
the distance function defined in this way satisfies the triangle-inequality.
For an element $e$ of $M$ and
a positive integer $d$, $N^M_d(e)$ denotes the set containing all elements of $M$ at distance at most $d$ from $e$,
i.e., those elements $f$ such that there is a circuit of $M$ containing both $e$ and $f$ of size at most $d$.
A class $\MM$ of matroids has {\em locally bounded branch-width}
if there exists a function $f:\NN\to\NN$ such that for every matroid $M\in\MM$,
every element $e$ of $M$ and every positive integer $d$,
the branch-width of the matroid $M$ restricted to the elements of $N^M_d(e)$ is at most $f(d)$.

\subsection{FOL and MSOL properties}
\label{subsect-logic}

A {\em monadic second order logic formula} $\psi$ for a graph property
is built up in the usual way from variables for vertices, edges,
vertex and edge sets, the equality and element-set signs $=$ and $\in$,
the logic connectives $\land$, $\lor$ and $\neg$,
the quantification over variables and the predicate
indicating whether a vertex is incident with an edge;
this predicate encodes the input graph. A {\em first order logic formula}
is a monadic second order logic formula with no variables for vertex and
edge sets. A formula with no free variables is called a {\em sentence}.
We call a graph property is {\em first order (monadic second order) logic property}
if there exists a first order (monadic second order) logic sentence
describing the property, respectively. We further abbreviate
first order logic to FOL and monadic second order logic to MSOL.

Similarly, we define first order logic formulas for matroids.
Such formulas contain variables for matroid elements and
the predicate indicating whether a set of elements is independent;
the predicate encodes the input matroid.
Monadic second order logic formulas contain variables for elements and
their sets and the predicate describing the independent set of matroids, and
first order logic formulas do not contain variables for sets of elements.
Analogously, first order and monadic second order logic matroid properties
are those that can be described by first order and monadic second order logic sentences.
An example of a first order logic property is the presence of a fixed matroid $M_0$ in an input matroid,
i.e., the existence of a restriction of the input matroid isomorphic to $M_0$.
For instance, if $M_0$ is the uniform matroid $U_{2,3}$ and $\indep_M$ the predicate
describing the independence in an input matroid $M$, then the corresponding
first order logic sentence is
$$\exists a,b,c; \indep_M(\{a,b\})\land \indep_M(\{a,c\})\land \indep_M(\{b,c\})\land\neg \indep_M(\{a,b,c\})\;\mbox{.}$$

As indicated in Table~\ref{table}, monadic second order logic graph properties
can be decided in linear time in classes of graphs with bounded tree-width~\cite{bib-courcelle90},
monadic second order logic matroid properties can be decided in cubic time
in classes of matroids with bounded branch-width that are given by a representation
over a fixed finite field~\cite{bib-hlineny06} and first order logic graph properties
can be decided in almost linear time in classes of graphs with locally bounded tree-width~\cite{bib-frick+}.
Theorem~\ref{thm-main} completes the picture by establishing that
first order logic matroid properties are fixed parameter tractable
in classes of oracle-given regular matroids with locally bounded branch-width.

The aspect of FOL properties exploited in~\cite{bib-frick+} to obtain their algorithmic results
is the locality of FOL properties. Let $P_1,\ldots,P_k$ be relations on a set $E$ and
let $r_i$ be the arity of $P_i$, $i=1,\ldots,k$.
The {\em Gaifman graph} for the relations $P_1,\ldots,P_k$ is the graph with vertex set $E$ such that
two elements $x_1$ and $x_2$ of $E$ are adjacent if there exist an index $i\in\{1,\ldots,k\}$ and
elements $x_3,\ldots,x_{r_k}$ of $E$ such that $(x_1,\ldots,x_{r_k})\in P_i$.
A formula $\psi[x]$ is {\em $r$-local} if all quantifiers in $\psi$ are restricted
to the elements at distance at most $r$ from $x$ in the Gaifman graph.
Gaifman established the following theorem which captures
locality of FOL properties.

\begin{theorem}[Gaifman~\cite{bib-gaifman}]
\label{thm-gaifman}
If $\psi_0$ is a first order logic sentence with predicates $P_1,\ldots,P_k$,
then the sentence $\psi_0$ is equivalent to a Boolean combination of sentences of the form
\begin{equation}
\exists x_1,\ldots,x_k \left(\bigwedge_{1\le i<j\le k}d(x_i,x_j)>2r \land \bigwedge_{1\le i\le k}\psi[x_i]\right)\label{eq-gaifman}
\end{equation}
where $d(x_i,x_j)$ is the distance between $x_i$ and $x_j$ in the Gaifman graph and
$\psi$ is an $r$-local first order logic formula for some integer $r$.
\end{theorem}

In particular, Theorem~\ref{thm-gaifman} reduces deciding FOL sentences to $r$-local formulas $\psi$.

\section{Local tree-width vs.~local branch-width}
\label{sect-ltw+lbw}

In this section, we relate graph classes with locally bounded tree-width and
matroid classes with locally bounded branch-width.

\begin{theorem}
\label{thm-ltw+lbw}
Let $\GG$ be a class of graphs and $\MM$ the class of cycle matroids
$M(G)$ for graphs $G\in\GG$. If the class $\GG$ has locally bounded tree-width,
then $\MM$ has locally bounded branch-width.
\end{theorem}

\begin{proof}
Assume that $\GG$ has locally bounded tree-width, and
let $f_G$ be the function for $\GG$ from the definition of a class of
locally bounded tree-width. We claim that the function $f_M(d):=f_G(\lfloor d/2\rfloor)+1$
certifies that $\MM$ has locally bounded branch-width.

Let $G$ be a graph of $\GG$ and $M=M(G)$.
Fix an element $e$ of $M$ and an integer $d\ge 1$.
We show that the branch-width of $M$ restricted to $N^M_d(e)$
is at most $f_M(d)$. Let $v$ and $w$ be the end vertices of $e$ in $G$.
If $f\in N^M_d(e)$, then the edges $e$ and $f$ are contained in a circuit
with at most $d$ elements of $M$. This circuit is
a cycle of $G$ of length at most $d$ and
thus the distance of both the end-points of $f$ from $v$
is at most $\lfloor d/2\rfloor$. In particular, both the end-vertices
of all the edges $f\in N^M_d(e)$ are contained in $N^G_{\lfloor d/2\rfloor}(v)$.

By the definition of a graph class of locally bounded tree-width,
the tree-width of the subgraph of $G$ induced by $N^G_{\lfloor d/2\rfloor}(v)$
is at most $f(\lfloor d/2\rfloor)$. Hence, the subgraph $H$ of $G$
comprised of the edges from $N^M_d(e)$ has tree-width at most
$f(\lfloor d/2\rfloor)$ as $H$ is contained in the subgraph of $G$
induced by $N^G_{\lfloor d/2\rfloor}(v)$. Since $M$ restricted
to $N^M_d(e)$ is the graphic matroid corresponding to $H$,
the branch-width of $M$ restricted to $N^M_d(e)$
is at most $f_G(\lfloor d/2\rfloor)+1$ by Proposition~\ref{prop-bw+tw}.
We conclude that $f_M$ certifies that $\MM$ is a class of matroids
with locally bounded branch-width.
\end{proof}

The converse of Theorem~\ref{thm-ltw+lbw} is not true;
this does however not harm the view of our results as a generalization
of the result of Frick and Grohe on graph classes with locally bounded tree-width
since our result apply to classes of graphic matroids corresponding to graph classes
with locally bounded tree-width (and, additionally, to several other matroid classes).

Let us give an example witnessing that the converse implication does not hold.
Consider a graph $G_k$ obtained in the following way: start with the complete
graph $K_k$ of order $k$, replace each edge with a path with $k$ edges and
add a vertex $v$ joined to all the vertices of the subdivision of $K_k$.
Clearly, the class of graphs $\GG=\{G_k,k\in\NN\}$ does not have locally bounded
tree-width since the subgraph induced by $N^G_1(v)$ in $G_k$ contains
$K_k$ as a minor and thus its tree-width is at least $k-1$.

Let $\MM$ be the class of graphic matroids of graphs $G_k$, and
set $f_M(d)=\max\{3,d+1\}$. We claim that the function $f_M$ witnesses that
the class $\MM$ has locally bounded branch-width.
Let $M=M(G_k)$ and fix an element $e$ of $M$ and an integer $d\ge 3$;
let $G_k$ be the graph corresponding to $M$.
If $k<d$, then the tree-width of $G_k$ is at most $k+1\le d$ and
thus the branch-width of $M$ is at most $d+1$ by Proposition~\ref{prop-bw+tw}.
In particular, the branch-width of $M$ restricted to any subset of its elements
is at most $d+1$.

We now assume that $k\ge d$. Let $w$ be a vertex incident with $e$ that
is different from the vertex $v$. Observe that any edge contained in $N^M_d(e)$
is incident with a vertex of the subdivision of $K_k$ at distance at most $d-2<k$ from $w$.
Since $k\ge d$, the edges in $N^M_d(e)$ contained in the subdivision of $K_k$
form no cycles in $G_k$ and thus they yield a subgraph of tree-width one.
Adding the extra vertex $v$ increases the tree-width by at most one and
thus the tree-width of the subgraph of $G_k$ formed by the edges of $N^M_d(e)$ is at most two;
consequently, the matroid $M$ restricted to $N^M_d(e)$ has branch-width at most three.

\section{Constructing Gaifman graph}
\label{sect-gaifman}

Let $\varphi$ be a first order logic sentence which contains the predicate $\indep_M$ encoding an input matroid $M$.
Further, let $d$ be the depth of quantification in $\varphi$.
Introduce new predicates $C^1_M,\ldots,C^d_M$ such that the arity of $C^k_M$ is $k$ and
$C^k_M(x_1,\ldots,x_k)$ is true if and only if $\{x_1,\ldots,x_k\}$ is a circuit of $M$.
For instance, $C^1(x_1)$ is true if and only if $x_1$ is a loop.
For an oracle-given matroid $M$, each predicate $C^k_M(x_1,\ldots,x_k)$, $k=1,\ldots,d$,
can be evaluated in linear time. The {\em circuit reduction} $\varphi^C$ of $\varphi$ is then the sentence
obtained from $\varphi$ by replacing every $\indep_M(x_1,\ldots,x_{\ell})$ in $\varphi$
by the conjunction:
$$\bigwedge_{1\le k\le\ell}\quad\bigwedge_{1\le j_1<\cdots<j_k\le\ell}\neg C^k_M(x_{j_1},\ldots,x_{j_k})\;\mbox{.}$$

Observe that a sentence $\varphi$ is satisfied for $M$ if and only if its circuit reduction $\varphi^C$ is.
In our algorithmic considerations, we prefer working with circuit reductions of sentences
since their Gaifman graphs are better-structured and, more importantly, distances in them are related
to distances in matroids.
Let $G^C_{M,d}$ be the graph with vertex set being the elements of $M$ and two elements of $M$
are adjacent if they are contained in a circuit of length at most $d$, i.e., their distance in $M$
is at most $d$.
Clearly,
the graph $G^C_{M,d}$ is the Gaifman graph for the circuit reduction of $\varphi$ and the matroid $M$
when $d$ is the depth of quantification of $\varphi$.
Proposition~\ref{prop-2circuits} yields the following (the bound $d\ell$ is not tight).

\begin{lemma}
\label{lm-gaifman-distance}
Let $\varphi$ be a first order logic sentence and $d$ an integer.
If two elements $x$ and $y$ in $M$ are at distance $\ell$ in the graph $G^C_{M,d}$,
then there exists a circuit of $M$ containing both $x$ and $y$ with length at most $d\ell$.
In particular, the distance of $x$ and $y$ in $M$ is at most $d\ell$.
\end{lemma}

In the rest of this section,
we focus on constructing the graph $G^C_{M,d}(\varphi)$ efficiently for an oracle-given regular matroid $M$ and
a first order logic sentence $\varphi$. To achieve our goal, we will utilize Seymour's Regular Matroid Decomposition Theorem.
To use this result, we show that the problem, which we call Minimum dependency weight circuit,
is fixed-parameter tractable for graphic and cographic matroids.
In this problem, elements of a matroid are assigned weights and we seek a circuit
containing given elements with minimum weight; however, if the circuit contains
two elements from one of special $3$-element circuits,
then the weight of the sought circuit gets adjusted (this we refer to as dependency).
The problem is formally defined as follows.
\begin{center}
\begin{tabular}{ll}
Problem: & {\bf Minimum dependency weight circuit (MDWC)} \\ \\
Parameter: & an integer $\ell$ \\ \\
Input: & a matroid $M$ with a set $F$ of one, two or three elements of $M$\\
       & a collection $\TT$ of disjoint $3$-element circuits of $M$\\
       & a positive integral function $w$ on $E(M)\setminus\bigcup_{T\in\TT}T$\\
       & non-negative integral functions $w_T$ on $2^T$, $T\in\TT$, such that\\
       & $w_T(\emptyset)=0$ and $w_T(A)\ge|A|$ for any non-empty $A\subseteq T$\\ \\
Output: & the minimum weight $w(C)$ of a circuit $C$ with $F\subseteq C$\\
        & where $w(C)=\sum\limits_{e\in C\setminus\bigcup\limits_{T\in\TT}T} w(e)+\sum\limits_{T\in\TT} w_T(C\cap T)$\\
	& if such minimum weight is at most $\ell$
\end{tabular}
\end{center}

To establish the fixed parameter tractability of MDWC problem for regular matroids,
we start with graphic matroids which are more straightforward to handle than cographic ones.

\begin{lemma}
\label{lm-cycle}
Minimum dependency weight circuit problem is fixed parameter tractable
in the class of graphic matroids.
\end{lemma}

\begin{proof}
The statement follows by a simple application of
the technique called color coding introduced in~\cite{bib-alon94}.
For completeness, we overview the method.

Let $G$ be the graph corresponding to an input matroid $M$ and
$f_1,\ldots,f_m$ edges corresponding to the elements of $F$.
Suppose that the vertices of $G$ are colored with $\ell$ colors and
our goal is to find a rainbow cycle $C$ containing all $f_1,\ldots,f_m$
with its weight $w(C)$ at most $\ell$. Here,
a rainbow cycle stands for a cycle with vertices of mutually distinct colors.
Such a cycle can be found in time $2^{O(\ell\log \ell)}O(|E(G)|^2)$ as
we explain in the next paragraphs.

Note that if the weight $w(C)$ of a cycle $C$ is at most $\ell$,
then the cycle has length at most $\ell$; this follows from the facts that
each element of $C$ not contained in a triple $T$ of $\TT$ contributes
at least one to the weight of $C$ and each triple $T$ contributes at least
$|C\cap T|$ since $w_T(C\cap T)\ge|C\cap T|$. So, it is enough to examine
cycles of length at most $\ell$.

Now, let us try to find a rainbow cycle of length exactly $k$, $2\le k\le\ell$.
Fix a cyclic sequence $\sigma$ of $k$ distinct colors.
We now find a cycle containing $f_1,\ldots,f_m$ with the colors of vertices given
by $\sigma$ of minimum weight (if such a cycle exists).
To find such a cycle,
first check whether the colors of the end-vertices of the edges of $F$
appear consecutively in $\sigma$; if not, such a cycle does not exist.
By symmetry, we can assume that $\sigma_1$ and $\sigma_k$ are colors of the end-vertices of $f_1$;
let $v_1$ be the end vertex of $f_1$ with the color $\sigma_1$ and $v_k$ the other end vertex.
Further, let $i_2$ ($i_3$) be the larger index of the color of an end vertex of $f_2$ ($f_3$, respectively).
If $m\le 2$, set $i_3$ to $\infty$; similarly, if $m=1$, set $i_2$ to $\infty$.

Let $V_i$ be the set of the vertices with color $\sigma_i$, $i=1,\ldots,k$ and set $V_0=V_k$.
Sequentially for $i=1,\ldots,k$, we compute for each edge $f$ between $V_{i-1}$ and $V_i$
the minimum weight of a path $P$ starting with the edge $f_1=v_kv_1$ that ends with the edge $f$
such that the vertices of $P$
have colors $\sigma_k,\sigma_1,\ldots,\sigma_i$ (in this order) and such that if $i\ge i_2$ ($i\ge i_3$),
the path $P$ includes $f_2$ ($f_3$, respectively). The weight of a path $P$ is defined naturally as
$$w_P=\sum_{e\in P\setminus\bigcup\limits_{T\in\TT}T} w(e)+\sum_{T\in\TT} w_T(P\cap T)\;\mbox{.}$$
Note that for some edges $f$ between $V_{i-1}$ and $V_i$ such a path $P$ may not exist.

For $i=1$, such a path exists only for $f=f_1$ and its weight is $w(f_1)$.
For $i>1$, the weights of such paths for edges between $V_{i-1}$ and $V_i$ can be computed
from the weights for edges between $V_{i-2}$ and $V_{i-1}$ in time $O(|E(G)|^2)$ including
weight adjustments in case that the consecutive edges are contained in one of the triples of $\TT$.
Note that such a path may contain at most two edges from any triple $T\in\TT$ and they must
be consecutive unless $k=3$ and the path is one of the triples of $\TT$.
Also note that if $i=i_j$, $j\in\{2,3\}$,
we ignore all edges $f$ between $V_{i-1}$ and $V_i$ different from $f_j$.

Finally, the minimum weight of a cycle containing all edges $f_1,\ldots,f_m$ with length $k$ and
vertices of colors $\sigma_1,\ldots,\sigma_k$ is computed based on the weights computed
for edges between $V_{k-1}$ and $V_k$. We discard all edges not leading to $v_k$ and
adjust weights if the edge between $V_{k-1}$ and $V_k$ and $f_1$ are contained together
in a triple of $\TT$.

Based on the above algorithm, we have shown that the Minimum dependency weight circuit problem
is fixed parameter tractable if
we can construct a small family of vertex colorings of $G$ such that every cycle of $G$
with length at most $\ell$ is rainbow in at least one of these colorings.
The above algorithm is applied for each coloring in the family and since every cycle
containing all edges $f_1,\ldots,f_m$ has length at most $\ell$, we eventually find a cycle
with minimum weight.

A desired family of vertex colorings of $G$ can be constructed using perfect hash functions.
A family $\FFF$ of hash functions from $\{1,\ldots,n\}$ to $\{1,\ldots,\ell\}$ is called {\em $\ell$-perfect}
if for every $A\subseteq\{1,\ldots,n\}$ with $|A|=\ell$, there exists $f\in\FFF$ mapping
the elements of $A$ to mutually distinct numbers. It is known~\cite{bib-schmidt90} that
there exist explicit families of $\ell$-perfect hash functions of size $2^{O(\ell)}\log^2 n$.
Since the vertex colorings of $G$ given by the functions forming an $\ell$-perfect family of hash functions
have the property that each cycle of length at most $\ell$ is rainbow in at least one of these colorings,
the statement of the lemma follows.
\end{proof}

Before we handle the case of cographic matroids, we need one more definition.
If $M$ is a cographic matroid, then a family $\CC$ of its circuits is {\em simple}
if there exists a graph $G$ corresponding to $M$ such that every circuit of $\CC$
corresponds to a cut around a vertex of $G$ (in particular, if $\CC$ is empty,
it is simple).

\begin{lemma}
\label{lm-cut}
Minimum dependency weight circuit problem is fixed parameter tractable
for cographic matroids providing that $\TT$ is simple.
\end{lemma}

\begin{proof}
We adopt the method of Kawarabayashi and Thorup~\cite{bib-cut-fpt} proving
that finding an edge-cut with at most $\ell$ edges and at least $k$ components
is fixed parameter tractable. Compared to their problem,
the cut we seek, besides weight computing issues,
is required to include the edges of $F$ and, more importantly, to be inclusion-wise minimal.

To solve Minimum dependency weight circuit problem, we will show how to solve
another problem which we call Weight dependent powercut. The problem is defined as follows.
\begin{center}
\begin{tabular}{ll}
Problem: & {\bf Weight dependent powercut} \\ \\
Parameter: & a positive integer $\ell$ \\ \\
Input: & a graph $G$ with {\bf non-negative} integer ranks $r_v$ of vertices, $v\in V(G)$\\
       & a set $T_0$ of vertices of $G$, called terminals, such that\\
       & $\mbox{}\qquad\qquad|T_0|+\sum\limits_{v\in V(G)}r_v\le 4\ell$\\
       & a collection $\TT$ of disjoint sets of one, two or three edges such that\\
       & for every $T\in\TT$, the edges of $T$ have a common end vertex $v(T)$\\
       & a positive integral function $w$ on $E(G)\setminus\bigcup_{T\in\TT}T$\\
       & non-negative integral functions $w_T$ on $\{1,2,\ldots,2^{r_{v(T)}}\}\times 2^T$, $T\in\TT$, \\
       & such that $w_T(i,A)\ge|A|$ for any $i$ and $A\subseteq T$\\ \\
Output: & for every equivalence relation $\rho$ on $T_0$,\\
        & for all choices $i_v\in\{1,\ldots,2^{r_v}\}$, $v\in V(G)$,\\
        & a cut $C$ with minimum weight $w(C)$\\
        & where $w(C)=\sum\limits_{e\in C\setminus\bigcup\limits_{T\in\TT}T} w(e)+\sum\limits_{T\in\TT} w_T(i_{v(T)},C\cap T)$ such that\\
	& every component of $G-C$ contains at least one terminal and\\
	& two terminals are in the same component iff they are $\rho$-equivalent\\
	& if the minimum weight of such cut $C$ is at most $\ell$
\end{tabular}
\end{center}
In the above problem as well as in the rest of the proof,
a {\em cut} in a graph $G$ is a set of edges of $G$ joining
$X$ and $V(G)\setminus X$ for some $X\subseteq V(G)$.
Note that the output of Weight dependent powercut problem is not a single cut
but a (large) family of cuts.

We will design a recursive procedure for computing weight dependent powercuts.
The running time of this procedure will be bounded by $O(\ell^{\ell^{\ell^{O(\ell)}}}m^3)$
where $m$ is the number of edges of the input graph.
To solve the problem stated in the lemma, we proceed as follows.
Let $G$ be the graph corresponding to the input matroid.
The rank of each vertex is chosen to be $0$ and
$T_0$ is set to be the end vertices of the edges corresponding to the elements of $F$.
The collection $\TT$ is identical to that in the input and the functions $w$ and $w_T$
also coincide (with setting the first parameter $1$ in $w_T$).
Out of the cuts returned by the procedure, we choose a minimum weight cut among those
corresponding to equivalence relations $\rho$ with two classes (and thus two components)
such that the end vertices of each edge of $F$ belong to different classes.
Clearly, this is the sought cut $C$.

Let us remark at this point that the ranks of vertices play an important role
in the recursive steps of our procedure: if we split along a cut which
separates some edges contained in the same set $T$ of $\TT$, the rank of
the vertex $v(T)$ will be increased to allow considering different possible
selections among edges not included in the considered side of the cut.
More details are given when describing the appropriate steps of the procedure.

We now describe the recursive procedure.
The input graph can be assumed to be connected: otherwise, we solve
the problem separately for each component and combine the obtained solutions together.
Observe that if a minimum cut separating two vertices $u$ and $v$ of $G$ has
at least $\ell+1$ edges, then the vertices $u$ and $v$ are always in the same component $G-C$
if $w(C)\le\ell$.
Hence, such vertices $u$ and $v$ can be identified. The resulting loops are removed,
causing some sets of $\TT$ to drop in size.
The rank of the new vertex is $r_u\cdot r_v$.
The weight function $w$ stays the same as well as
the functions $w_T$ for $T$ with $v(T)\not\in\{u,v\}$. For $T$ with $v(T)\in\{u,v\}$,
we define
$$w'_T(i,A):=\left\{\begin{array}{cl}
             w_T(i\mod 2^{r_u},A) & \mbox{if v(T)=u, and}\\
	     w_T(i\div 2^{r_u},A) & \mbox{if v(T)=v.}\end{array}\right.$$
Here, we use $a\div b:=\lceil a/b\rceil$ and $a\mod b:=a-((a\div b)-1)\cdot b$.
Clearly, this change in the input does not affect the structure of cuts with weight at most $\ell$ and
thus the output of the problem.
So, we may assume that for every pair of vertices $u$ and $v$,
there is a cut with at most $\ell$ edges seperating $u$ and $v$.

Observe that the number of output cuts can be at most $(4\ell)^{4\ell}$ since
there are at most $|T_0|^{|T_0|}\le (4\ell)^{|T_0|}$ distinct equivalence relations on $T_0$ and
at most $2^{\sum_{v\in V(G)} r_v}\le 2^{4\ell-|T_0|}$ choices of $i_v$, $v\in V(G)$.
These cuts can contain together at most $s:=\ell(4\ell)^{4\ell}$ edges.

A vertex of $G$ is said to be {\em big} if its degree is at least $D=2s+\ell$.
We distinguish three cases in the procedure.
\begin{itemize}
\item {\bf The graph $G$ contains at least two big vertices.} Let $u_1$ and $u_2$ be two big vertices.
      Let $C$ be a cut 
      with minimum number of edges
      separating $u_1$ and $u_2$ and let $G_1$ and $G_2$ be the two components of $G-C$.
      Since $u_1$ and $u_2$ are big, each of the components $G_1$ and $G_2$ contains at least $2s$ edges.
      By the assumption on the input, it holds that
      $$|V(G_1)\cap T_0|+\sum_{v\in V(G_1)}r_v\le 2\ell\quad\mbox{ or }\quad|V(G_2)\cap T_0|+\sum_{v\in V(G_2)}r_v\le 2\ell\;\mbox{.}$$
      By symmetry, let the former be the case.
      
      We now define an instance of Weight dependent powercut problem with the graph $G_1$.
      The new set of terminals is formed by $V(G_1)\cap T_0$ and the end vertices of the edges of $C$.
      The ranks of the vertices are equal to the their original ranks with the exception of the end vertices
      of the edges of $C$ whose rank is increased by one for each incident edge of $C$ contained in a set $T$
      of $\TT$ with at least one edge of $G_1$.
      Observe that the sum of the number of terminals and the vertex ranks does not exceed $4\ell$.

      The edge weights (the function $w$) are preserved. The functions $w_T$ are also the same for $T\in\TT$
      if $v(T)$ is not incident with an edge of $C$.

      Consider now a vertex $v_0$ incident with an edge of $C$.
      Let $r$ be the original and $r'$ the new rank of $v_0$.
      For $T\in\TT$ with $v(T)=v_0$ and $T\cap C=\emptyset$,
      define
      $$w'_T(i,A):=w_T(i\mod 2^r,A)\mbox{ for every $A\in 2^T$ and $i\in\{1,\ldots,2^{r'}\}$.}$$
      Let $T_1,\ldots,T_{r'-r}$ be the sets of $\TT$ such that $T_k\cap C\not=\emptyset$,,
      $T_k\cap E(G_1)\not=\emptyset$ and $v(T_k)=v_0$ (we allow multiplicites if
      $|T_k\cap C|=2$, but assume that they always appear consecutively).
      If $|T_k\cap C|=1$ then define
      $$w'_{T_k}(i,A):=\left\{\begin{array}{ll}
                       w_{T_k}(i\mod 2^r,A) & \mbox{if $i\div 2^{r+k-1}$ is even and}\\
                       w_{T_k}(i\mod 2^r,A\cup\{e_k\}) & \mbox{if $i\div 2^{r+k-1}$ is odd.}
		       \end{array}\right.$$
      If $|T_k\cap C|=2$ and $T_k=T_{k+1}$, then define
      $$w'_{T_k}(i,A):=\left\{\begin{array}{ll}
                       w_{T_k}(i\mod 2^r,A) & \mbox{if $i\div 2^{r+k-1}\mod 4=4$,}\\
                       w_{T_k}(i\mod 2^r,A\cup\{e_k\}) & \mbox{if $i\div 2^{r+k-1}\mod 4=1$,}\\
                       w_{T_k}(i\mod 2^r,A\cup\{e_{k+1}\}) & \mbox{if $i\div 2^{r+k-1}\mod 4=2$,}\\
                       w_{T_k}(i\mod 2^r,A\cup\{e_k,e_{k+1}\}) & \mbox{otherwise.}
		       \end{array}\right.$$
      We now solve Weight dependent powercut problem in $G_1$ with respect to the above defined parameters.

      By the definition of Weight dependent powercut problem,
      at least one of the optimal solutions of Weight dependent powercut problem of $G$ would be preserved
      if we contracted the edges of $G_1$ not contained in any of the cuts of the solution of the problem for $G_1$.
      After contracting such edges of $G_1$, the number of edges of $G$ is decreased by at least $|E(G_1)|-s\ge s$ and
      Weight dependent powercut problem is solved in the new instance. The contraction of each edge
      is performed by identifying its end vertices in the way described earlier.
\item {\bf The graph $G$ contains no big vertices.}
      If $G$ has less than $D^{2s+6}$ edges, we solve Weight dependent powercut problem by brute force.
      Otherwise, construct greedily a vertex-minimal connected induced subgraph $H_0$ of $G$
      containing at least $h=2(s+1)^2$ edges. Observe that $H_0$ has at most $h+D\le 2h$ edges since
      adding a single vertex to a subgraph can increase the number of edges by at most $D$.
      We now construct subgraphs $H_1,\ldots,H_{s+1}$. Suppose that we have constructed $H_i$.
      The subgraph $H_{i+1}$ is a vertex-minimal subgraph of $G\setminus (V(H_0)\cup\cdots V(H_i))$
      containing all neighbors of the vertices of $H_i$ such that each component of $H_{i+1}$
      has at least $D$ edges or it is a component of $G\setminus (V(H_0)\cup\cdots V(H_i))$.

      The components of $H_i$ with at least $D$ edges are called {\em huge} and the subgraph of $H_i$
      formed by them is denoted by $H'_i$.
      It can be shown by induction that the number of vertices of $H_i$ does not exceed $2hD^{2i}$.
      This implies (since $G$ is connected and has at least $D^{2s+5}$ vertices) that
      every subgraph $H_i$ has at least one huge component; in particular, each $H'_i$ is non-empty.

      For every subgraph $H'_i$, $i=1,\ldots,s+1$, we check whether any two (huge) components of $H'_i$
      are separated by a cut with at most $\ell$ edges in $G$. If so, consider the side of this cut
      with the number of terminals and the sum of ranks not exceeding $2\ell$ and
      process it recursively as in the first case. Then, contract
      all edges of this side not contained in any of the cuts of the obtained solution and process the resulting graph.
      Hence, we can assume that no two components of $H'_i$ can be separated by a cut of
      weight at most $\ell$ edges.

      Let $G_i$, $i=1\ldots,s+1$, be the graph obtained from the subgraph induced by the vertices
      $V(H_0)\cup\cdots\cup V(H_i)$ by identifying all vertices of $H'_i$ to a single vertex $v_i$ and
      consider the instance of Weight dependent powercut problem for $G_i$ obtained
      with the set of terminals $\{v_i\}\cup (T\cap V(G_i))$ where the ranks of vertices of $G_i$ are
      the same as in $G$ except for the vertex $v_i$. The rank of $v_i$ reflects (in the same way as
      when splitting along an edge cut with at most $\ell$ edges) that some edges
      between $H_i$ and the rest of the graph may be in sets of $\TT$; in particular, the rank of $v_i$
      is bounded from above
      by the sum of the ranks of vertices inside $H_i$ and the number of edges of leaving $H_i$.
      Though the sum of the number of terminals and the vertex ranks can (substantially) exceed $4\ell$,
      Weight dependent powercut problem can be solved by brute force since the number of vertices of $G_i$
      is at most $2h+2hD^2+\cdots+2hD^{2i}\le 4hD^{2i}$. We do so for every $i=1\ldots,s+1$;
      let $E_i$ be the set of edges contained in at least one of the cuts in the solution of the problem for $i$.

      Consider a solution of the original instance of Weight dependent powercut problem.
      If this solution does not use any edge inside $H'_i$, the edges not contained in $E_i$
      can be contracted without affecting the solution.
      Since the number of edges in the solution is at most $s$, it avoids edges inside at least one
      of the graphs $H'_i$, $i=1,\ldots,s+1$. So, we can contract in $G$ the edges of $H_0$
      not contained in $E_1\cup\cdots\cup E_{s+1}$.
      This decreases the number of edges of $G$ by at least $h-(s+1)^2\ge (s+1)^2$ and
      we recursively solve the new instance of Weight dependent powercut problem.
\item {\bf The graph $G$ contains one big vertex.}
      This case is handled similarly to the previous case. Let $b$ be the only big vertex of $G$.
      The induced subgraphs $H_0,\ldots,H_{s+1}$ are constructed as in the previous case
      but in $G\setminus b$ instead of $G$. The component of $H_i$ is huge if it contains
      at least $D$ edges or it contains a vertex adjacent to $b$; $H'_i$ is the induced subgraph of $G$
      formed by the vertex $b$ and the vertices of the huge components of $H_i$.
      
      From now on,
      the procedure runs in the same way as in the previous case.
      First, every pair of the components of $H'_i$ is checked
      for the existence of an edge cut with at most $\ell$ edges (and if one is found,
      the case is handled recursively).
      Next, the edge sets $E_1,\ldots,E_{s+1}$ in graphs $G_1,\ldots,G_{s+1}$ are computed and
      the edges of $H_0$ not contained in $E_1\cup\cdots\cup E_{s+1}$ are contracted.
      The new instance is solved recursively.
\end{itemize}
The running time of the algorithm when an edge cut with at most $\ell$ edges is found is bounded as follows:
at least half of the edges of the graph in the first call are discarded, so the overall
running time can be estimated by the sum of the running times of the two recursive calls.
The running time in the second and the third cases when no small edge cut is found is bounded
by $D^{O(s)}=s^{O(s)}$; since each step decreases the number of edges of the input graph
by at least $s^2$, the number of these steps does not exceed $O(m)$. Each step requires time $O(m^2)$
as there can be $O(m)$ pairs of huge components tested for the existence of a small edge cut,
each test requiring $O(m)$ time (we are looking for edge cuts of with at most $\ell$ edges).
Hence, the total running time of the algorithm to solve Weight dependent powercut problem
is $O(\ell^{\ell^{\ell^{O(\ell)}}}m^2)$.
\end{proof}

We can now prove the main theorem of this section.

\begin{theorem}
\label{thm-circuit}
The problem of deciding the existence of a circuit of length at most $d$
containing two given elements of a regular matroid is fixed parameter
tractable when parameterized by $d$.
\end{theorem}

\begin{proof}
We solely focus on presenting the algorithm witnessing the fixed parameter
tractability of the problem without trying to optimize its running time.
Let $M$ be the input matroid and $f_1$ and $f_2$ the two given elements.

If the matroid $M$ is not connected and the two given elements do not lie
in the same component, then there is no circuit containing both of them.
Otherwise, if the input matroid is not connected, we can restrict
the input matroid to the component containing the given element(s).
So, we assume the input matroid is connected.

We apply Theorem~\ref{thm-regular} to obtain the series of $2$-sums and $3$-sums of
graphic matroids, cographic matroids and matroids isomorphic to $R_{10}$ that yields
the input matroid. Let $M_1,\ldots,M_K$ be these matroids. By symmetry,
we assume that the element $f_1$ is contained in $M_1$.
Also fix graphs $G_i$ corresponding to graphic and cographic matroids $M_i$
appearing in this sequence (such graphs can be constructed in polynomial time
by Theorem~\ref{thm-graphicness}).

If we allow some of the $M_1,\ldots,M_K$ to be isomorphic to matroids obtained
from the matroid $R_{10}$ by adding elements parallel to its original elements,
we can also assume the following:
if one of $3$-sums yielding the input matroid $M$ involves a circuit $E$,
then there exist indices $i_1$ and $i_2$ such that $E\subseteq E(M_{i_j})$ for $j=1,2$.
This assumption allows us to create a tree with $K$ vertices representing
the way in which the matroid $M$ is obtained from $M_1,\ldots,M_K$. More precisely,
there exists a rooted tree $T_M$ with a root $u_1$ and additional $K$-1 nodes $u_2,\ldots,u_K$
such that:
\begin{itemize}
\item the nodes $u_1,\ldots,u_K$ of $T_M$ one-to-one correspond to the matroids $M_1,\ldots$, $M_K$, and
\item for every pair of $3$-element circuit $E$ appearing in a $3$-sum of the construction of $M$
      is contained in matroids $M_i$ and $M_j$ such that $u_i$ is a child of $u_j$.
\end{itemize}
Moreover, this tree $T_M$ can be constructed in polynomial time given the sequence
$M_1$, $\ldots$, $M_K$ with operations yielding the matroid $M$,

For a matroid $M_i$ corresponding to a node $u_i$ of $T_M$,
we call a set of $M_i$ corresponding to a sum with the matroid of a child of $u_i$
a {\em special set}.
For $i\not=1$, we call the set of $M_i$ corresponding to the sum with the matroid of its parent
a {\em parent set}. For $i=1$, the parent set is defined to be $\{f_1\}$.
By a simple extension of the tree $T_M$ with additional nodes, one can achieve that
the special $3$-element sets of cographic matroids $M_i$ correspond to cuts around vertices in $G_i$.
So, we assume this to be the case in what follows.
Similarly, adding a $2$-sum with the matroid $U_{1,2}$ allows us to assume that the element $f_2$
is contained in a matroid corresponding to a leaf node of $T_M$.

Define $N_1,\ldots,N_K$ to be the matroids obtained based on $T_M$
by the following recursive construction from the leaves towards the root.
If $u_i$ is a leaf node of $T_M$, set $N_i$ to the matroid $M_i$.
If $u_i$ is a node with children $u_{j_1},\ldots,u_{j_k}$,
set $N_i$ to the matroid obtained from $M_i$ and $N_{j_1},\ldots,N_{j_k}$
by sums along the special sets of $M_i$.
Note that the matroid $N_1$ is the input matroid $M$.

We will compute the following quantities $m_i$, $m_{i,j}$ and $m'_{i,j}$ for the matroids $N_1,\ldots,N_K$:
\begin{itemize}
\item if $N_i$ does not contain the element $f_2$ and its parent set has a single element $e_1$,
      we compute the minimum number $m_i$ of elements of a circuit $C$ of $N_i$ containing $e_1$,
\item if $N_i$ contains $f_2$ and its parent set has a single element $e_1$,
      we compute the minimum number $m_i$ of elements of a circuit $C$ of $N_i$ containing both $e_1$ and $f_2$,
\item if $N_i$ does not contain the element $f_2$ and its parent set has three elements $\{e_1,e_2,e_3\}$,
      for each $j=1,2,3$, we compute the minimum number $m_{i,j}$ of elements of a circuit $C$ of $N_i$ 
      with $C\cap\{e_1,e_2,e_3\}=\{e_j\}$ and
      the minimum number $m'_{i,j}$ of elements of a circuit $C$ of $N_i$ with $C\cap\{e_1,e_2,e_3\}=\{e_j\}$
      such that $C\triangle\{e_1,e_2,e_3\}$ is also a circuit, and
\item if $N_i$ contains $f_2$ and its parent set has three elements $\{e_1,e_2,e_3\}$,
      for each $j=1,2,3$, we compute the minimum number $m_{i,j}$ of elements of a circuit $C$ of $N_i$ 
      with $C\cap\{e_1,e_2,e_3,f_2\}=\{e_j,f_2\}$ and
      the minimum number $m'_{i,j}$ of elements of a circuit $C$ of $N_i$ with $C\cap\{e_1,e_2,e_3,f_2\}=\{e_j,f_2\}$
      such that $C\triangle\{e_1,e_2,e_3\}$ is also a circuit.
\end{itemize}
All the minimums above are computed assuming that the number of elements of $C$ does not exceed $d$;
if it exceeds this quantity, we do not determine the value.
Note that $m_1$ is the minimum number of
elements of a circuit of $N_1=M$ containing both $f_1$ and $f_2$.
So, determining these quantities in a fixed parameter way will complete the proof.

The numbers $m_i$, $m_{i,j}$ and $m'_{i,j}$ are computed from the leaves towards the root.
Consider a node $u_i$ of $T_M$.
If the node corresponds to the matroid $U_{1,2}$ containing $f_2$, set $m_i=2$.
In the general case, we proceed as follows.
Let $\TT$ be the set of special $3$-element sets of $M_i$.
Let $T=\{t_1,t_2,t_3\}\in\TT$ and let $N_k$ be the matroid with the parent set $T$;
note that the node corresponding to $N_k$ is a child of $u_i$.
Define
$$w_T(\emptyset)=0,\quad w_T(\{t_j\})=m'_{k,j}-1,$$
$$w_T(T\setminus\{t_j\})=\max\{2,m_{k,j}-1\}\mbox{  and  } w_T(T)=d+1\;\mbox{.}$$
Next, we assign weights to the elements of $M_i$ not contained in a triple of $\TT$.
The elements of $M_i$ contained in a $1$-element special set summed with a parent set of $N_k$
are assigned weight $m_k-1$ and
the elements of $M_i$ not contained in a special set are assigned weights one. This defines the weight function $w$.
Further, if $N_i$ contains $f_2$, then $u_i$ has a child $u_{i'}$ such that $N_{i'}$ contains $f_2$.
Fix such $i'$ (if it exists) for what follows. 

We show that the sought quantities can be computed using the Minimum dependency weight circuit problem.
Suppose first that the parent set of $N_i$ has a single element $e_1$.
A simple manipulation using binary representations of $N_i$ and $M_i$ and
the definition of $N_i$ yields the following.
If $N_i$ does not contain the element $f_2$,
the quantity $m_i$ is the minimum weight of a circuit of $M_i$
with respect to the weights $w$ and $w_T$, $T\in\TT$, containing $e_1$.
If $N_i$ contains $f_2$,
then the quantity $m_i$ is the minimum weight of a circuit of $M_i$ containing $e_1$ and
containing at least one element of the parent set of $N_{i'}$. Since there are at most three
such elements, we get at most three instances of MDWC problem.

If the parent set of $N_i$ has three elements, say, $e_1$, $e_2$ and $e_3$,
the situation gets more complicated.
Again, the fact that $N_i$ and $M_i$ are binary and the way in which $N_i$ is obtained from $M_i$
gives the following. The quantity $m_{i,j}$ is
the minimum weight of a circuit of $M_i$ containing $e_j$ and,
if $f_2\in N_i$, at least one element of the parent set of $N_k$.
The quantity $m'_{i,j}$ is the minimum weight of a circuit of $M_i$ containing
$e_{j+1}$ and $e_{j+2}$ and, if $f_2\in N_i$, at least one element of the parent set of $N_{i'}$.
We obtain at most eighteen instances of MDWC problem.

If $M_i$ is graphic or cographic, MDWC problem can be solved using Lemmas~\ref{lm-cycle} and
\ref{lm-cut}. Note that in the cographic case all $3$-element special sets correspond to cuts
around vertices by our assumptions. If $M_i$ is a matroid obtained from $R_{10}$ by adding
elements parallel to its original elements, we can just enumerate all circuits of $M_i$ (the number
of them is bounded by a polynomial of degree $D$ where $D$ is the maximum length of a circuit of $R_{10}$) and
we solve MDWC problem by brute force.
In all three cases, the exponent in the polynomial bounding the running time of the algorithm is constant.
\end{proof}

Since two elements of a matroid $M$ are adjacent in the graph $G^C_{M,d}$
if and only if they are contained in a common circuit of size at most $d$,
Theorem~\ref{thm-circuit} immediately yields the following.

\begin{corollary}
\label{cor-gaifman}
The graph $G^C_{M,d}$ can be constructed in polynomial time for an oracle-given regular matroid $M$ and
an integer $d$ and the degree of the polynomial in the estimate on its running time is independent of $d$.
\end{corollary}

\section{Deciding first order logic properties}
\label{sect-alg}

Since first order logic formulas are special cases of monadic second order logic formulas,
the results of Hlin{\v e}n{\'y}~\cite{bib-hlineny03,bib-hlineny06} imply the following.

\begin{lemma}
\label{lm-local-sentence}
Let $\MM$ be a class of regular matroids with locally bounded branch-width.
There exists a cubic-time algorithm that given an oracle-given $M\in\MM$,
an element $x$ of $M$, the graph $G^C_{M,d}$ and an $r$-local formula $\psi[x]$
with predicates $C^1_M,\ldots,C^d_M$ (the locality of $\psi$ is measured in $G^C_{M,d}$)
decides whether $\psi[x]$ is satisfied in $M$.
\end{lemma}

\begin{proof}
Let $f_{\MM}$ be the function witnessing that $\MM$ is
the class of matroids with locally bounded branch-width.
We now briefly describe the algorithm.
Given the graph $G^C_{M,d}$, the matroid $M$ and the element $x$,
construct the set $X$ of elements of $M$ at distance at most $r$ from $x$ in $G^C_{M,d}$.
Let $M_X$ be the matroid $M$ restricted to $X$. The matroid $M_X$ is binary (since $M$ is regular) and
thus we can construct in quadratic-time its representation over the binary field.
Moreover, the branch-width of $M_X$ is at most $f_{\MM}(rd)$.

Deciding whether $\psi[x]$ is satisfied now reduces to deciding whether this formula is satisfied in $M_X$.
However, $M_X$ is a matroid represented over the binary field and it has branch-width at most $f_{\MM}(rd)$.
Hence, this can be decided in cubic time using the techniques from~\cite{bib-hlineny03,bib-hlineny06}.
To be more precise, the results from~\cite{bib-hlineny03,bib-hlineny06} asserts that MSOL properties,
which include FOL properties, can be decided in cubic time for matroids with bounded branch-width that
are represented over a fixed finite field. Here, we have a constant $x$ appearing in our logic language and
a distance constraint
but it is straightforward to verify that the techniques from~\cite{bib-hlineny03,bib-hlineny06}
apply to this setting, too.
\end{proof}

We are now ready to prove the main result of this paper. The final part of the argument
is analogous to that used by Frick and Grohe in~\cite{bib-frick+}.

\begin{theorem}
\label{thm-main}
Let $\psi_0$ be a first order logic sentence and $\MM$ a class of regular matroids with locally bounded
branch-width. There exists a polynomial-time algorithm that decides whether an oracle-given
$n$-element matroid $M\in\MM$ satisfies $\psi_0$. Moreover, the degree of the polynomial in the estimate
on its running time is independent of $\psi_0$ and $\MM$, i.e., testing first order logic properties
in classes of regular matroids with locally bounded branch-width is fixed parameter tractable.
\end{theorem}

\begin{proof}
Let us start with describing the algorithm. Set $d$ to be the depth of quantification in $\psi_0$.
First, using Corollary~\ref{cor-gaifman}, construct the graph $G^C_{M,d}$. This can be done
in polynomial time where the degree of the polynomial in the time estimate is independent of $d$.
By Theorem~\ref{thm-gaifman}, the circuit reduction $\psi_0^C$ of $\psi_0$ is equivalent to a Boolean
combination of sentences of the form (\ref{eq-gaifman}). Hence, it is enough to show how
a sentence of the form (\ref{eq-gaifman}) can be decided in polynomial time.

Fix a sentence of the form (\ref{eq-gaifman}) where $\psi$ is an $r$-local formula
with predicates $C^1_M,\ldots,C^d_M$. By Lemma~\ref{lm-local-sentence}, we can test whether $\psi[x]$
is satisfied for every element $x$ of $M$ in cubic time. Let $X_\psi$ be the set of all elements $x$
such that $\psi[x]$ is satisfied. This set can be constructed in time $O(n^4)$ where $n$ is the number of
elements of $M$.

We now have to decide whether $X_\psi$ contains $k$ elements at pairwise distance at least $2r+1$ in $G^C_{M,d}$.
Choose $x_1\in X_\psi$ arbitrary.
If $x_1,\ldots,x_{i-1}$ have already been chosen,
choose $x_i$ to be an arbitrary element of $X_\psi$ at distance at least $2r+1$
from each of the elements $x_1,\ldots,x_{i-1}$ (if such $x_i$ exists).
Let $\ell$ be the number of elements obtained in this way.
Observe that the list $x_1,\ldots,x_{\ell}$ can be constructed in time $O(n^3)$
given $X_\psi$ and $G^C_{M,d}$.

If $\ell\ge k$, the sentence (\ref{eq-gaifman}) is satisfied. Assume that $\ell<k$.
Let $\xi$ be the smallest positive integer such that no two elements of $X_\psi$
have distance in $G^C_{M,d}$ between $8^{\xi}r+1$ and $8^{\xi+1}r$ (inclusively).
Since the number of pairs of elements in $x_1,\ldots,x_\ell$ is ${\ell\choose 2}\le {k\choose 2}$,
the number of different distances between the elements of $X_\psi$ is at most $k^2/2$ and
thus $\xi\le k^2$ (it can be shown that $\xi\le k$ but the quadratic bound suffice for our purposes).
Let $X'_\psi$ be a maximal subset of elements $x_1,\ldots,x_{\ell}$ such that
the distance between any two elements of $X'_\psi$ is at least $8^{\xi}r+1$.
Further, let $\ell'=|X'_\psi|$.
Since $k$ is bounded,
the integer $\xi$ and the set $X'_{\psi}$ can be determined in quadratic time.

For simplicity, let us assume that $X'_{\psi}=\{x_1,\ldots,x_{\ell'}\}$.
Let $A_i$, $1\le i\le\ell'$, be the set of elements of $M$ at distance at most $8^{\xi}r+2r$ from $x_i$ in $G^C_{M,d}$.
Note that if elements $a\in A_i$ and $a'\in A_{i'}$, $1\le i,i'\le\ell'$, are at distance at most $2r$,
then $i=i'$ (in particular, the sets $A_1,\ldots,A_{\ell'}$ are disjoint).
Indeed, if the distance of $a$ and $a'$ were at most $2r$,
then the distance between $x_i$ and $x_{i'}$
would be at most $2\cdot (8^{\xi}r+2r)+2r=2\cdot 8^{\xi}r+6r\le 8^{\xi+1}r$ by the triangle inequality,
which would contradict the choice of $\xi$.
On the other hand, every element of $X_\psi$ is at distance at most $2r$ from an element $x_j$, $1\le j\le\ell$, and
the element $x_j$ is at distance at most $8^{\xi}r$ from one of the elements $x_1,\ldots,x_{\ell'}$
by the choice of $X'_{\psi}$. Hence, every element of $X_\psi$ is contained in one of the sets $A_1,\ldots,A_{\ell'}$.

Since $X_\psi\subseteq A_1\cup\cdots\cup A_{\ell'}$ and
no two elements contained in distinct sets $A_1,\ldots,A_{\ell'}$ are at distance at most $2r$,
the maximum size of a subset of elements $X_\psi$ at pairwise distance at least $2r+1$
is equal to the sum of maximum sizes of such subsets in $A_1\cap X_\psi,\ldots,A_{\ell'}\cap X_\psi$.

Let $B_i$, $1\le i\le\ell'$, be the set of elements of $M$ at distance at most $2dr$ in $M$ from an element of $A_i$.
Observe that two elements of $A_i$ are at distance at most $2r$ in $G^C_{M,d}$ if and only if
they are at distance at most $2r$ in $G^C_{M_i,d}$ where $M_i$ is the matroid $M$ restricted to $B_i$.
In particular, the maximum size of a subset of $A_i\cap X_\psi$
formed by elements at pairwise distance at least $2r+1$ in $G^C_{M,d}$
is equal to the maximum size of a subset of $A_i\cap X_\psi$
formed by elements at pairwise distance at least $2r+1$ in $G^C_{M_i,d}$.
The matroid $M_i$ is formed by elements at distance at most $(8^{\xi}r+2r)d+2dr\le 8^{k^2}rd+4rd$ from $x_i$ and
thus its branch-width is at most $f_{\MM}(8^{k^2}rd+4rd)$
where $f_{\MM}$ is the function witnessing that the class $\MM$ has locally bounded branch-width.
Since it can be encoded in the first order logic that
the distance between two elements in $G^C_{M_i,d}$ is at most $2r$,
we can use the results of Hlin{\v e}n{\'y} on decidability of MSOL properties from~\cite{bib-hlineny03,bib-hlineny06}
to determine (in cubic time) whether there exists a subset of $A_i\cap X_\psi$ of a fixed size $m$, $1\le m\le k$,
formed by elements at distance at least $2r+1$ in $G^C_{M_i,d}$.

If there exists such a subset of $A_i\cap X_\psi$ of size $k$ for some $i$, $1\le i\le\ell'$,
then the formula (\ref{eq-gaifman}) is true.
Otherwise, let $m_i$ be the maximum size of such a subset of $A_i\cap X_\psi$, $1\le i\le\ell'$.
As we have explained, the numbers $m_i$ can be determined in cubic time (recall that $k$ does not depend on the input).
The formula (\ref{eq-gaifman}) is true if and only if $m_1+\cdots+m_{\ell'}\ge k$.
This finishes the description of the algorithm.
The correctness and polynomiality of the algorithm
where the degree of the polynomial does not depend on $\psi_0$ and $\MM$ follows from the arguments we have given.
\end{proof}

\section*{Acknowledgement}

The authors would like to thank Ken-ichi Kawarabayashi for discussing the details
of his fixed parameter algorithm for computing multiway cuts given in~\cite{bib-cut-fpt}
during the NII workshop---Graph Algorithms and Combinatorial Optimization.

\end{document}